\documentclass[oneside,english]{amsart}
\pdfoutput=1
\usepackage[T1]{fontenc}
\usepackage[latin9]{inputenc}
\usepackage{geometry}
\usepackage{varioref}
\usepackage{refstyle}
\usepackage{amsthm}
\usepackage{amstext}
\usepackage{amssymb}

\makeatletter

\AtBeginDocument{\providecommand\eqref[1]{\ref{eq:#1}}}
\AtBeginDocument{\providecommand\thmref[1]{\ref{thm:#1}}}
\AtBeginDocument{\providecommand\exaref[1]{\ref{exa:#1}}}
\AtBeginDocument{\providecommand\subref[1]{\ref{sub:#1}}}
\AtBeginDocument{\providecommand\fnref[1]{\ref{fn:#1}}}
\providecommand{\tabularnewline}{\\}
\RS@ifundefined{subref}
  {\def\RSsubtxt{section~}\newref{sub}{name = \RSsubtxt}}
  {}
\RS@ifundefined{thmref}
  {\def\RSthmtxt{theorem~}\newref{thm}{name = \RSthmtxt}}
  {}
\RS@ifundefined{lemref}
  {\def\RSlemtxt{lemma~}\newref{lem}{name = \RSlemtxt}}
  {}

\numberwithin{equation}{section}
\numberwithin{figure}{section}
\theoremstyle{plain}
\newtheorem{thm}{\protect\theoremname}
  \theoremstyle{definition}
  \newtheorem{example}[thm]{\protect\examplename}
  \theoremstyle{remark}
  \newtheorem*{rem*}{\protect\remarkname}
  \theoremstyle{plain}
  \newtheorem{fact}[thm]{\protect\factname}
  \theoremstyle{definition}
  \newtheorem{defn}[thm]{\protect\definitionname}
  \theoremstyle{remark}
  \newtheorem{claim}[thm]{\protect\claimname}

\usepackage[all,knot]{xy} 
\xyoption{knot}
\xyoption{frame}
\xyoption{arc}
\xyoption{curve}
\xyoption{poly}
\usepackage{pdfsync}
\makeatother

\usepackage{babel}
  \providecommand{\claimname}{Claim}
  \providecommand{\definitionname}{Definition}
  \providecommand{\examplename}{Example}
  \providecommand{\factname}{Fact}
  \providecommand{\remarkname}{Remark}
\providecommand{\theoremname}{Theorem}

\begin{document}

\title{Toroidal Spin-Networks: Towards a Generalization of the Decomposition
Theorem}

\author{Hans-Christian Ruiz}

\maketitle
\begin{center}
{\footnotesize Max-Planck-Institute for Gravitational Physics, Golm,
Germany (Internship)}
\par\end{center}{\footnotesize \par}

\begin{center}
{\footnotesize Email: }\textit{\footnotesize hans.christian.ruiz@aei.mpg.de}
\par\end{center}{\footnotesize \par}
\begin{abstract}
In this paper Moussouris' algorithm for the decomposition of spin
networks is reviewed and the implicit assumptions made in the Decomposition
Theorem relating a spin network with its state sum are examined. It
is found that the theorem in the original form hides the importance
of the orientation of the vertices in the spin networks and, more
important, that the algorithm for the evaluation of spin networks
assumes a cycle for the reduction of the graph to work in the proof
of the theorem. It is shown that this is not always the case and this
rises doubt about the generality of the theorem. Having this issue
in mind, the theorem is restated to account for toroidal spin networks,
i.e. networks cellular embeddable in the torus. For this, the minimal
non-planar spin network is examined and the algorithm is extended
to account for the non-planarity of toroidal spin networks. Furthermore,
three types of minimal non-planar spin networks are found, two of
them are toroidal and the third is cellular embeddable only in the
double torus and without a cycle. Some examples for both toroidal
spin networks are given and the relation between these is found. Finally,
the issues and the possibility of generalizing the Decomposition Theorem
is shortly discussed.
\end{abstract}

\section{Introduction}

Moussouris' algorithm for the evaluation of planar spin networks and
the Decomposition Theorem are important tools for the evaluation of
spin networks, for instance one uses them in order to extract the
state sum of a triangulated 3-dimensional manifold with boundary giving
rise to the Ponzano-Regge theory relating the 1-skeleton of the triangulation
to the ``partition function'' of the geometry of the given manifold.
This is made by considering the triangulation of the boundary as the
dual to a spin network and decomposing the manifold, by means of Moussouris'
algorithm, into tetrahedra, which are related in a specific way to
the $6j$-symbols found in the coupling theory of $SU(2)$-representations.
This association allows the construction of the state sum of the given
manifold as a sum over a product of $6j$-symbols, \cite{moussouris1983quantum,ponzano1969semiclassical,regge1961general}.
This idea can be generalized to a more abstract setting as described
in \cite{barrett1996invariants}, for instance the model in \cite{turaev1992state},
where the category of representations of the quantum group $U_{q}(\mathfrak{sl}_{2})$
is used, rather than the category of $SU(2)$-representations as in
the Ponzano-Regge model.

When considering general spin networks, one has to rise the question
about the global topology of the manifold encoded in the spin network
under consideration and whether is possible to blindly apply the Decomposition
Theorem to obtain a state sum encoding this information. In other
words, how is the global topology of the manifold reflected in the
evaluation of the corresponding spin network? This question is important
since from the graph-theoretical point of view planar spin networks,
i.e. those embeddable in a sphere, are a special case and one has
to imagine more complicated generic spin networks if one takes them
to be quantum states of 3-dimensional space. In what follows we will
explore this issue using some basic concepts of (topological) graph
theory taken from \cite{hartsfield2003pearls,beineke1997graph} to
describe the embeddings of spin networks in surfaces, in particular
the cellular embeddings of non-planar graphs with a $(3,3)$-bipartite
graph as subgraph, and present an extension of Moussouris' algorithm
in order to account for the non-planarity of the minimal spin network
with non-trivial topology, i.e. non-trivial topology of the surface
in which it is cellular embeddable.

The reader familiarized with graph theory and the evaluation of spin
networks via Moussouris' algorithm might want to take a look at the
example \vref{exa:case-3-3} and the following ones and jump directly
to \vref{sub:The-Toroidal-SN}. However, to make this paper self-contained
we will give some basic notions of graph theory with especial attention
to Kuratowski's Theorem stating an easy condition for a graph to be
non-planar, and the Rotation Scheme Theorem describing how to obtain
all cellular embeddings for a given graph out of a simple rule. Furthermore,
we will consider the $(3,3)$-bipartite graph $K_{3,3}$ as the minimal
non-planar spin network, construct its non-equivalent cellular embeddings
and show that these are the only possible ones up to permutation of
its vertices and their relative orientation among those vertices of
the same set in $K_{3,3}$. 

Moreover, Moussouris' algorithm for the evaluation of planar spin
networks and the Decomposition Theorem are presented in section \vref{sub:The-Decomposition-Theorem}
of this paper. We will apply and extend these ideas to evaluate the
above mentioned spin networks, which will allow us to define a toroidal
symbol in order to attempt a generalization of the Decomposition Theorem.
To achieve this a few identities for relating the evaluation of the
different embeddings in the torus are given. We explain the main concepts
involving the generalization of the evaluation of non-planar spin
networks in terms of toroidal symbols, however, the process for obtaining
the main result is still ongoing, thus some caveat concerning the
generalization of Kuratowski's Theorem and its possible solution will
be pointed out.

\section{Kuratowski's Theorem and the Embedding of Graphs in Surfaces\label{sec:Kuratowski's-Theorem-Embeddings}}

A \textbf{graph} $G$ is a pair of sets $(V_{G};E_{G})$ where $V_{G}\neq\emptyset$
is the set of vertices of $G$ and $E_{G}$ is a set of unordered
pairs of elements of $V_{G}$ which might be empty. The pairs of vertices
are called edges and they are defined, in an abstract way, as a relation
between the vertices in $G$. If two vertices form an edge, we say
they are adjacent or neighbors. A \textbf{subgraph} of a graph $G$
is a graph $H=(V_{H};E_{H})$ such that $V_{H}\subseteq V_{G}$ and
$E_{H}\subseteq E_{G}$, \cite{hartsfield2003pearls}. Another concept
related to subgraphs are the \textbf{minors} of a graph; these are
graphs obtained from $G$ by a succession of edge-deletions and edge-contractions.
If the minor $M$ was obtained only by edge contractions, then $G$
is said to be contractible to $M$, \cite{beineke1997graph}.

We are able to consider only trivalent vertices since spin networks
of higher degree, i.e. with vertices of higher valence, are expandable
to cubic graphs. On the other hand, a spin network is only allowed
to have at most one edge between two adjacent vertices. If two vertices
are joined by more than one edge, the structure is called a multigraph.
In the case of cubic graphs we only have trivalent vertices, hence,
a multigraph would have two vertices joined by at most three edges.
This impose, however, no constraints in the class of spin networks
since the double edge can be reduced to a single edge using Schur's
lemma and a triple edge is exactly the theta-net defined as the value
of a 3-vertex.

Two graphs are \textbf{homeomorphic} if they can be obtained from
the same graph by subdividing its edges. Subdividing an edge $e=vw$
between two vertices $v,w$ is the operation of inserting a new vertex
$z$ such that $e$ is replaced by two new edges $vz$ and $zw$,
\cite{beineke1997graph}. 

An embedding of a graph on a surface is called \textbf{cellular }if
each region is homeomorphic to an open disc. It is in this sense that
planarity is defined, namely, a graph is planar if it can be embedded
in a plane, hence, in the 2-sphere. Surprisingly, there is a simple
criterion for determining whether a graph is planar, or not, given
by Kuratowski's theorem.
\begin{thm}
\textbf{\label{thm:Kuratowski's-Theorem}Kuratowski's Theorem.} A
graph is planar if and only if it has neither $K_{5}$ nor $K_{3,3}$
as a minor, i.e. there are no subgraphs homeomorphic or contractible
to $K_{5}$ and $K_{3,3}$.
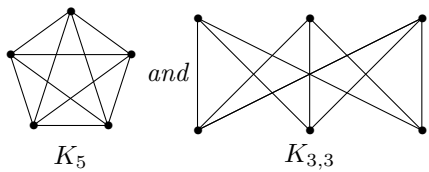
\[
\xy/r2.0pc/:{\xypolygon5"C"{}},"C1"*@{*};"C3"*@{*}**@{-},"C1";"C4"**@{-},"C2"*@{*};"C4"*@{*}**@{-},"C2";"C5"*@{*}**@{-},"C3";"C5"**@{-},"C0"-(0,1.3)*{K_{5}}\endxy\:\text{ and }\xy/r2.0pc/:{\xypolygon4"B"{~:{(2.5,0):(0,.5)::}~>{}}},"B1";"B2"**{}?<>(.5)="B5"*@{*},"B3";"B4"**{}?<>(.5)="B6"*@{*},"B1"*@{*},"B2"*@{*},"B3"*@{*},"B4"*@{*},"B1";"B3"**@{-},"B1";"B4"**@{-},"B1";"B6"**@{-},"B5";"B3"**@{-},"B5";"B6"**@{-},"B5";"B4"**@{-},"B1";"B3"**@{-},"B2";"B3"**@{-},"B2";"B6"**@{-},"B2";"B4"**@{-},"B0"-(0,1.3)*{K_{3,3}}\endxy
\]

\end{thm}
The graph $K_{5}$ is called the complete graph on five vertices and
$K_{3,3}$ is the $(3,3)$-bipartite graph. Notice that the $K_{5}$
graph is 4-valent and can be expanded as a spin network to obtain
the Petersen graph which is a cubic graph with 10 vertices as depicted
below. If one deletes any of its vertices and the edges incident to
it, one finds a graph homeomorphic to the $(3,3)$-bipartite graph.
Hence, we only need to focus on the latter graph.

\medskip{}

\begin{center}
\xy/r3.0pc/:
{\xypolygon5"C"{}},
"C1"*@{*};"C3"*@{*}**@{-},"C1";"C4"**@{-},"C2"*@{*};"C4"*@{*}**@{-},"C2";"C5"*@{*}**@{-},"C3";"C5"**@{-},"C0"-(0,1.3)*{K_{5}}
\endxy
$\rightarrow$
\xy/l1pc/:
{\xypolygon5"A"{}},
{\xypolygon5"B"{~:{(-2.5,0):}~>{}}},
{\xypolygon5"C"{~:{(-3.75,0):}}},
"B1"*@{*};"C1"*@{*}**@{-},"B2"*@{*};"C2"*@{*}**@{-},"B3"*@{*};"C3"*@{*}**@{-},"B4"*@{*};"C4"*@{*}**@{-},"B5"*@{*};"C5"*@{*}**@{-},
"B1";"A3"**@{-},"B1";"A4"**@{-},"B2";"A5"**@{-},"B2";"A4"**@{-},"B3";"A5"**@{-},"B3";"A1"**@{-},"B4";"A1"**@{-},"B4";"A2"**@{-},
"B5";"A3"**@{-},"B5";"A2"**@{-},
"A0"+(0,3.9)*{\text{Petersen graph}}
\endxy
\par\end{center}

\medskip{}

Closed surfaces, on the other hand, are categorized into orientable
and non-orientable surfaces, e.g. the sphere $S^{2}$, torus $T^{2}$
or the real projective plane $P^{2}$. Any oriented surface is homeomorphic
to the sphere or to the connected sum $T^{2}\#T^{2}\#\dots\#T^{2}$
of a finite number $h$ of tori. Here we will only consider the embeddings
of spin networks in orientable closed surfaces.

Whether a given graph $G$ is embeddable in an orientable surface
$S_{h}$ or not depends on its genus $\gamma(G)$, which is defined
to be the minimum genus of any orientable surface in which $G$ is
embeddable, i.e. $G$ is embeddable in $S_{h}$ if $h\geq\gamma(G)$.
In fact, any graph can be embedded in a surface with enough handles
just by adding a handle at each crossing, but we are rather interested
in cellular embeddings%
\footnote{In fact, if $\gamma(G)=h$, then every embedding of $G$ on $S_{h}$
is cellular, \cite{beineke1997graph}.%
} for which Euler's formula hold,
\begin{equation}
v-e+f=2-2h\label{eq:Eulers-formula-for-graph-embeddings}
\end{equation}
where $v,e,f$ are the number of vertices, edges and faces respectively
and $h$ is the (orientable) genus of $S_{h}$. In this context, a
planar graph has genus $\gamma(G_{planar})=0$.

There is no general formula for calculating the orientable genus of
a given graph, however, for the $(s,r)$-bipartite graph it is given
by
\begin{equation}
\gamma(K_{s,r})=\left\lceil \frac{(r-2)(s-2)}{4}\right\rceil \label{eq:genus-of-graph}
\end{equation}
where $\left\lceil x\right\rceil $ denotes the next integer bigger
than $x$. Hence, $K_{3,3}$ is embeddable in the torus but not in
the sphere since $\gamma(K_{3,3})=\left\lceil \frac{1}{4}\right\rceil =1$,
\cite{beineke1997graph}.

Notice that from the \eqref{Eulers-formula-for-graph-embeddings}
there is a topological constraint to the allowed cellular embeddings
for a given graph. For instance, $K_{3,3}$ has 6 vertices and 9 edges,
thus, we obtain a constraint for the number of faces, $f=5-2h$, since
$f,h>0$. Furthermore, from \eqref{genus-of-graph} we have $h\geq1$,
hence, $f=3$ or $f=1$ in the case where $K_{3,3}$ is embedded%
\footnote{From now on, whenever we refer to an embedding, it is meant a cellular
embedding.%
} in $T^{2}\text{ or }T^{2}\#T^{2}$ respectively.

Is there other information encoded in the graph that can help us to
obtain all possible embeddings? Does the orientation of the vertices
impose a constraint on the embedding? Now, having found the number
of possible faces (or 2-cells) for the embeddings, we want to construct
oriented surfaces such that the cellular embeddings are automatically
realized. In order to achieve this, we need to find the circuits of
the given graph $G$ and attach 2-cells to the regions bounded by
them. 

A \textbf{circuit} is a closed walk%
\footnote{A walk in $G$ is an alternating sequence $v_{1}e_{1}v_{2}e_{2}\dots v_{n-1}e_{n-1}v_{n}$
of vertices and edges of $G$, where every edge $e_{i}$ is incident
with $v_{i}\text{ and }v_{i+1}$, and $v_{i}\neq v_{i+1}$. If $v_{1}=v_{n}$,
then it is a closed walk, \cite{hartsfield2003pearls}. It is important
to notice here that this definition is a property of the graph itself.
We will, however, abuse the use of the language and refer to the regions
bounded by the circuits (cf. \thmref{Rotation-Scheme-Theorem}) as
\textit{embedded circuits}, when they have repeated edges in both
directions, or as \textit{embedded cycles} when no edges are repeated.
To clarify, the difference is that the latter concepts are related
to the embedding of the graph in a surface and the definition given
above is a property of the graph related only indirectly to the embedding
of the graph through \thmref{Rotation-Scheme-Theorem}.%
} in $G$ such that no edge is repeated in the same direction. One
may imagine a circuit as walking along an edge in a certain direction
and, when getting to a vertex, the direction to follow (either clockwise
or counter-clockwise) is given by the orientation, also called rotation,
of the vertex under consideration. Thus, a given configuration of
the orientations of all vertices in the graph, called a rotation scheme,
induces a set of circuits giving rise to a specific cellular embedding.
Hence, all embeddings can also be described by giving the orientation%
\footnote{The orientation is given as a cyclic permutation of the neighbors
encountered while going clockwise around the vertices. In this convention,
the orientation of a vertex $v$ is the equivalent class of even permutations,
in the case of positive orientation denoted $(+1)$, or odd permutations,
in the case of negative orientation denoted $(-1)$, of its neighbors.
Such a definition takes into account the fact that both directions,
anticlockwise and clockwise, can be described by listing the neighbors
in order of their appearance when going \textit{clockwise}, for instance,
if $(abc)$ describes the positive orientation of $\xygraph{!{0;/r1.5pc/:}[u]!{\xcapv@(0)>{a}}*{\bullet}[lr(0.1)]!{\sbendv@(0)>{b}}[ll]!{\sbendh-@(0)>{c}}}$,
then $(acb)$ describes the negative orientation, meaning going around
counter-clockwise, which can also be regarded as permutating the edges
$vc$ and $vb$ giving $\xygraph{!{0;/r1.5pc/:}[u]!{\xcapv@(0)>{a}}*{\bullet}[lr(0.1)]!{\sbendv@(0)>{c}}[ll]!{\sbendh-@(0)>{b}}}$.
Notice that \textit{in the diagram} the listing of neighbors of $v$
is still clockwise.%
} of each vertex and specifying the regions bounded by the circuits
obtained from applying the so called \textbf{rotation rule}: after
the edge $xy$, take the edge $yz$, where $z$ is the successor to
$x$ in the permutation of the neighbors of the vertex $y$, see examples
below. The previous discussion is formalized in the next theorem,
\cite{beineke1997graph}.
\begin{thm}
\textbf{\label{thm:Rotation-Scheme-Theorem}Rotation Scheme Theorem.
}Let $G$ be a connected graph with v vertices and e edges, and let
$\Pi=\{\pi_{1},\pi_{2},\dots,\pi_{v}\}$ be a set of cyclic permutations
of the neighbors of the vertices $\{1,2,\dots,v\}$, i.e. a set of
a given orientation of all vertices. Let $W_{1},W_{2},\dots,W_{f}$
be the circuits obtained by applying the rotation rule to $\Pi$.
Then the circuits are the boundaries of the regions of a cellular
embedding of $G$ in $S_{h}$, with $h=(2-v+e-f)/2$, the genus of
the orientable surface $S_{h}$. Hence, all possible cellular embeddings
of a graph are provided by the rotation schemes.
\end{thm}
Let us consider the graph $K_{3,3}$ with $v=6,\, e=9$. This is a
graph with trivalent vertices, hence, there are two possible rotations%
\footnote{In fact, if a vertex $v$ has degree $d$, then there are $(d-1)!$
different rotations of $v$.%
} for each of the six vertices. As a consequence, there are $2^{6}=64$
different sets $\Pi_{i=1,\dots,64}$. The question that arises immediately
is whether all these sets induce topological inequivalent embeddings
or not. In other words, if we disregard the labeling of the vertices,
how many different embeddings of $K_{3,3}$ in the torus or the double-torus
exist?

Before making a general claim, let us work out some examples to illustrate
the construction of embeddings by the rotation rule in order to understand
the relation between the set of orientations and the embeddings induced
by them. We will achieve this by listing the vertices and their neighbors
and from this list extract the circuits in the embedding using the
rotation rule, \cite{beineke1997graph,hartsfield2003pearls}. 

Notice that for a general $(r,s)$-bipartite graph the defining characteristics
are: \textit{(i)} The $r$ vertices corresponding to a set, say $R$,
are adjacent to each of the $s$ vertices corresponding to another
set, say $S$, and \textit{(ii)} $R\cap S=\emptyset$. In the case
of $K_{3,3}$ each of the odd vertices, $R=\{1,3,5\}$, are adjacent
to each of the even vertices, $S=\{2,4,6\}$. Thus, we denote (up
to cyclic permutations) positive orientations of even and odd vertices
as $(135)\text{ and }(246)$ respectively, and $(153)\text{ and }(264)$
for negative orientations. For instance, a standard way of picturing
$K_{3,3}$ with minimal crossing is

\medskip{}

\begin{center}
\def\objectstyle{\scriptscriptstyle}
\xy/r2pc/:
{\xypolygon4"L"{~:{(1,0):(0,3.5)::}~>{}{\bullet}}},
"L1";"L4"**{}?<>(.5)="L5"*{\bullet},"L2";"L3"**{}?<>(.5)="L6"*{\bullet},
"L2";"L1"**@{-},"L1";"L6"**@{-},"L6";"L5"**@{-},"L5";"L3"**@{-},"L3";"L4"**@{-},
"L4";"L2" **\crv{(1,-4) & (-4,-4.5)&(-1,2.5)},
"L4";"L6" **\crv{(.8,-3.5) & (-2,-2.5) & (-1,-.5)},
"L2";"L5" **@{-},
"L1";"L3" **\crv{(4,0)},
"L2"+(0,.2)*{1},"L1"+(0,.2)*{2},"L6"+(-.1,.2)*{3},"L5"+(.1,.2)*{4},"L3"+(0,.2)*{5},"L4"+(0,.2)*{6},
"L2"-(.2,.25)*{(246)},"L1"+(.4,0)*{(153)},"L6"+(0,-.3)*{(246)},"L5"+(.25,-.2)*{(153)},"L3"+(-.1,-.2)*{(264)},"L4"+(.4,-.1)*{(135)}
\endxy
\par\end{center}

\medskip{}

This configuration has orientations $(246)$ for the vertices $1$
and $3$, $(264)$ for vertex $2$ and $(153)$ for the vertices $2$
and $4$, $(135)$ for vertex $6$. 
\begin{example}
\label{exa:case-3-3}Consider the case where the vertices in each
set have the same orientation, i.e. either $(+1,+1,+1)$ or $(-1,-1,-1)$.
For instance, the vertices $\{1,3,5\}$ have negative orientation
while the vertices $\{2,4,6\}$ have all positive orientation:\\

\begin{tabular}{|c|c|c|c|c|c|c|}
\hline 
Vertex & 1 & 2 & 3 & 4 & 5 & 6\tabularnewline
\hline 
\hline 
Neighbors/Orientation & (264) & (135) & (264) & (135) & (264) & (135)\tabularnewline
\hline 
\end{tabular}\\

From this information we extract the circuits which will help us to
construct the embedding corresponding to this configuration. For instance,
take the edge $(12)$ and apply the rotation rule on it, i.e. the
neighbor of $2$ coming after $1$ in the cyclic permutation $(135)$
is $3$, hence, the next edge in the walk is $(23)$. Apply again
the rule to get $(36)$ and so on. After some steps, depending on
how long the walk is, one gets to the edge where the procedure started,
meaning that one has to stop and apply the same procedure to another
edge different than the ones encountered in the previous walk. In
this specific configuration, this algorithm results in the following
disjoint circuits,\end{example}
\begin{enumerate}
\item $(12)\rightarrow(23)\rightarrow(36)\rightarrow(65)\rightarrow(54)\rightarrow(41)\rightarrow(12)$;
\item $(25)\rightarrow(56)\rightarrow(61)\rightarrow(14)\rightarrow(43)\rightarrow(32)\rightarrow(25)$;
\item $(21)\rightarrow(16)\rightarrow(63)\rightarrow(34)\rightarrow(45)\rightarrow(52)\rightarrow(21)$. 
\end{enumerate}
Notice that there is a difference between the ``side'' $(xy)$ and
$(yx)$ reflecting the direction of the walk, hence, there are 18
``sides'' available to build the circuits. In this case, there are
three regions with six sides as boundaries, thus, the embedding has
three faces and corresponds to an embedding in the torus as depicted
below:

\medskip{}

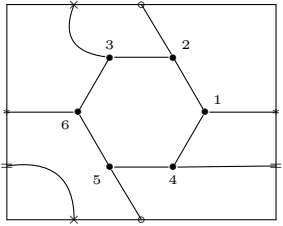
\begin{center}
\def\objectstyle{\scriptscriptstyle}
\xy/r2pc/:
{\xypolygon4"T"{~:{(3,0):(0,.8)::}}},{\xypolygon6"A"{~:{(1,0):(0,1)::}{\bullet}}},
"T1";"T2"**{}?<>(.5)="T5"*{\circ},"T4";"T3"**{}?<>(.5)="T6"*{\circ},
"T1";"T4"**{}?<>(.5)="T7"*{*},"T2";"T3"**{}?<>(.5)="T8"*{*},
"A1";"T7"**@{-},"A4";"T8"**@{-},"A2";"T5"**@{-},"A5";"T6"**@{-},
"T7";"T4"**{}?<>(.5)="T9"*{=},"T8";"T3"**{}?<>(.5)="T10"*{=},
"T5";"T2"**{}?<>(.5)="T11"*{\times},"T6";"T3"**{}?<>(.5)="T12"*{\times},
"T9";"A6"**@{-},"T10";"T12"**\crv{"T12"+(0,1)},"T11";"A3"**\crv{"T11"-(.3,.75)},
"A1"+(.2,.2)*{1},"A2"+(.2,.2)*{2},"A3"+(0,.2)*{3},"A4"-(.2,.2)*{6},"A5"-(.2,.2)*{5},"A6"-(0,.2)*{4}
\endxy
\par\end{center}

\medskip{}

\begin{example}
\label{exa:case-1-1}Consider the case where one vertex has the opposite
orientation relative to the two other vertices of the same set. For
instance, the case where vertex $1$ and $2$ have positive orientation,
$(246)$ and $(135)$ respectively, and the rest have negative orientation:\\

\begin{tabular}{|c|c|c|c|c|c|c|}
\hline 
Vertex & 1 & 2 & 3 & 4 & 5 & 6\tabularnewline
\hline 
\hline 
Neighbors/Orientation & (246) & (135) & (264) & (153) & (264) & (153)\tabularnewline
\hline 
\end{tabular}\\

Using the rotation rule we obtain the following circuits,\end{example}
\begin{enumerate}
\item $(12)\rightarrow(23)\rightarrow(36)\rightarrow(61)\rightarrow(12)$;
\item $(21)\rightarrow(14)\rightarrow(45)\rightarrow(52)\rightarrow(21)$;
\item $(32)\rightarrow(25)\rightarrow(56)\rightarrow(63)\rightarrow(34)\rightarrow(41)\rightarrow(16)\rightarrow(65)\rightarrow(54)\rightarrow(43)\rightarrow(32).$
\end{enumerate}
Notice that this time, we obtain two circuits of length four and a
single one of length ten. Hence, the 18 sides available form three
faces and the embedding is in a torus:

\medskip{}

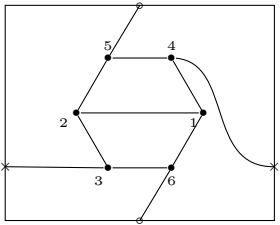
\begin{center}
\def\objectstyle{\scriptscriptstyle}
\xy/r2pc/:
{\xypolygon4"T"{~:{(3,0):(0,.8)::}}},{\xypolygon6"A"{~:{(1,0):(0,1)::}{\bullet}}},
"T1";"T2"**{}?<>(.5)="T5"*{\circ},"T4";"T3"**{}?<>(.5)="T6"*{\circ},
"T1";"T4"**{}?<>(.5)="T7","T2";"T3"**{}?<>(.5)="T8",
"A1";"A4"**@{-},"A3";"T5"**@{-},"A6";"T6"**@{-},
"T7";"T4"**{}?<>(.5)="T9"*{\times},"T8";"T3"**{}?<>(.5)="T10"*{\times},
"T10";"A5"**@{-},"T9";"A2"**\crv{"A6"+(.5,0)&"A2"+(1,0)},
"A1"-(.15,.15)*{1},"A2"+(0,.2)*{4},"A3"+(0,.2)*{5},"A4"-(.2,.15)*{2},"A5"-(.15,.2)*{3},"A6"-(0,.2)*{6}
\endxy
\par\end{center}

\medskip{}

\begin{example}
\label{exa:case-3-1}Finally, consider the case where the orientation
of all vertices in one set is the same while in the other set we have
one vertex with the opposite orientation relative to the other two
vertices. For instance, the case where all odd vertices have orientation
$(246)$ and vertex $2$ has positive orientation as well, while the
vertices $4$ and $6$ have orientation $(153)$:\\

\begin{tabular}{|c|c|c|c|c|c|c|}
\hline 
Vertex & 1 & 2 & 3 & 4 & 5 & 6\tabularnewline
\hline 
\hline 
Neighbors/Orientation & (246) & (135) & (246) & (153) & (246) & (153)\tabularnewline
\hline 
\end{tabular}\\

In this case we obtain, after using the described algorithm, only
one circuit with 18 sides, 
\begin{multline*}
(12)\rightarrow(23)\rightarrow(34)\rightarrow(41)\rightarrow(16)\rightarrow(65)\rightarrow(52)\rightarrow(21)\rightarrow(14)\rightarrow(45)\rightarrow(56)\rightarrow\dots\\
\dots\rightarrow(63)\rightarrow(32)\rightarrow(25)\rightarrow(54)\rightarrow(43)\rightarrow(36)\rightarrow(61)\rightarrow(12)
\end{multline*}
where all sides are walked exactly once%
\footnote{Notice that a circuit induced in this way may have repeated vertices
and edges used in both directions, however, if the edge is repeated
in the same direction the algorithm must stop, \cite{hartsfield2003pearls}.%
}. This configuration thus corresponds to an embedding in $T^{2}\#T^{2}$
which can be represented in a plane in a similar manner as the torus,
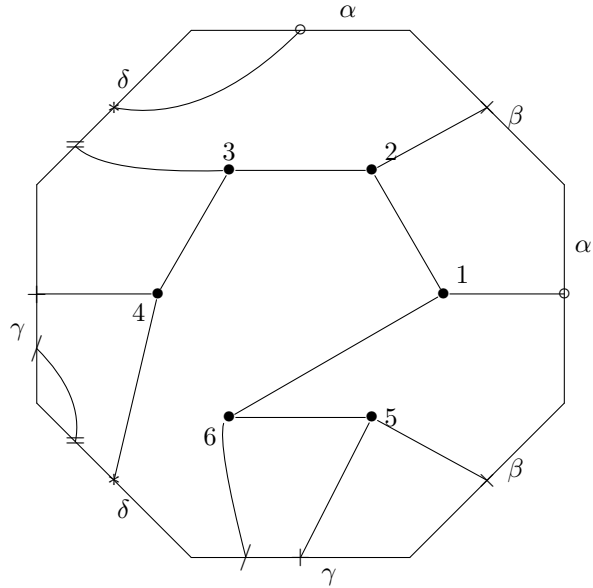
\[
\xy/r3pc/:{\xypolygon8"T"{~:{(3,0):(0,1)::}}},{\xypolygon6"A"{~:{(1.5,0):(0,1)::}~>{}{\bullet}}},"T1";"T2"**{}?<>(.5)="T9"*{\times}+(.3,-.1)*{\beta},"T2";"T3"**{}?<>(.5)="T10"*{\circ}+(.5,.2)*{\alpha},"T3";"T4"**{}?<>(.5)="T11"*{*}+(.1,.3)*{\delta},"T4";"T5"**{}?<>(.5)="T12"*{+}+(-.2,-.4)*{\gamma},"T5";"T6"**{}?<>(.5)="T13"*{*}+(.1,-.3)*{\delta},"T6";"T7"**{}?<>(.5)="T14"*{+}+(.3,-.2)*{\gamma},"T7";"T8"**{}?<>(.5)="T15"*{\times}+(.3,.1)*{\beta},"T8";"T1"**{}?<>(.5)="T16"*{\circ}+(.2,.5)*{\alpha},"A1"+(.2,.2)*{1},"A2"+(.2,.2)*{2},"A3"+(0,.2)*{3},"A4"-(.2,.2)*{4},"A5"-(.2,.2)*{6},"A6"+(.2,0)*{5},"T13";"A4"**@{-},"A4";"A3"**@{-},"A3";"A2"**@{-},"A2";"A1"**@{-},"A1";"A5"**@{-},"A5";"A6"**@{-},"A6";"T15"**@{-},"A2";"T9"**@{-},"A1";"T16"**@{-},"T10";"T11"**\crv{"T10"-(1,1)},"A6";"T14"**@{-},"A4";"T12"**@{-},"T11";"T4"**{}?<>(.5)="T17"*{=},"T12";"T5"**{}?<>(.5)="T18"*{/},"T5";"T13"**{}?<>(.5)="T19"*{=},"T6";"T14"**{}?<>(.5)="T20"*{/},"A3";"T17"**\crv{"T17"+(.3,-.3)},"T18";"T19"**\crv{"T18"+(.5,-.5)},"T20";"A5"**\crv{"A5"-(.2,0)},\endxy
\]
where the Greek letters denote the borders of the frame that have
to be glued together in order to obtain a double torus and the symbols
on them identify corresponding points.
\end{example}
Now, define the \textbf{value} $\mathfrak{v}$\textbf{ of a set of
vertices} as the modulus of the sum of orientations $\pm1$ of the
vertices in that set, e.g. the value of $\{1,3,5\}$ with orientation
$(-1-1-1)$ is $|-3|=3$. The definition is such that, if the orientations
of all vertices in a given set change, then the value remains invariant.
For instance, one can achieve a change of orientation of all vertices
in, say, the set $R=\{1,3,5\}$, e.g. $(+1+1-1)$, by an odd permutation
of the vertices in $S=\{2,4,6\}$ such that $(+1+1-1)\mapsto(-1-1+1)$
but $\mathfrak{v}_{++-}=|+1|=|-1|=\mathfrak{v}_{--+}$. In fact, by
permutations of the vertices in a set, one can construct all equivalent
diagrams%
\footnote{One has to consider the operation $R\rightleftarrows S$ as well.%
}, i.e. giving the same embedding, since these operations do not change
the 2-cells of the embedding, it merely results in a permutation of
the vertices in it.

Observe that the three cases in the examples above are the only cases
possible if we consider only the \textit{relative} orientation between
vertices of the same set, in which case the value is either $\mathfrak{v}=3$,
when all vertices have the same orientation, or $\mathfrak{v}=1$,
when one of the vertices has the opposite orientation relative to
the other two in the same set. The value of each set is independent
of each other, hence, we have the cases where the pair of values are
$(3,3)$, $(3,1)$, $(1,3)$ and $(1,1)$. However, since $K_{3,3}$
is symmetric under exchange of sets $R\leftrightarrow S$ preserving
the orientation, we can regard $(1,3)$ and $(3,1)$ as equivalent
cases.

Another way of looking at this is to consider the partition of 18
in summands with some constraints. Each of the summands represents
a circuit and their value represent the length of the circuit. The
defining characteristics of the $(3,3)$-bipartite graph do not allow
the construction of circuits with an odd number of edges since this
would mean that two vertices of the same set are adjacent. Thus, the
partition of 18 cannot contain any odd numbers. It follows that the
smallest possible circuit has length 4. Furthermore, the only number
of summands in the partition can be 1 or 3 since they represent the
regions of the embeddings. From these restrictions we conclude that
the only partitions of 18 allowed are $18,\;6+6+6,\;4+4+10$ and $4+8+6$.
However, the latter partition is not realized. To see this, notice
that it is not possible to construct a 4-circuit which is not a 4-cycle%
\footnote{A \textbf{cycle} is a circuit which does not have any repeated edges
in any direction.%
} since this would mean that either one edge is repeated, in which
case the edge would need to have two loops at each vertex, or two
edges are repeated, this is not possible since all vertices in the
graph are 3-valent. Thus, the only possibility is to have a 4-cycle;
this implies automatically the existence of another 4-cycle. To see
this, notice that the 4-cycle has two vertices of each set, therefore,
there are two more vertices available to construct the graph, one
of each kind. The defining characteristics of $K_{3,3}$ impose the
constraint that these two vertices must be adjacent to each other
and to the corresponding vertices in the original 4-cycle. This leaves
no other possibility but to construct another 4-cycle, in contradiction
to the partition $4+8+6$.
\begin{rem*}
Observe the symmetry of the bipartite graph under permutation of its
vertices reflected in the pair of values of the sets as well as in
the partition of 18, thus, we may call the $(3,3)$ case ``maximal
symmetric'', the $(1,1)$ case ``minimal symmetric'' and the $(3,1)=(1,3)$
case ``asymmetric''. Therefore, we can think of this pair of values
as the ``degree of symmetry'' of the graph.
\end{rem*}
From the examples \exaref{case-3-3}, \exaref{case-1-1} and \exaref{case-3-1}
we see that the partitions $[6+6+6],\;[4+4+10]$ and $18$ correspond
to the pair of values $(3,3),\:(1,1)$ and $(1,3)$ respectively.
We say that the maximal (minimal) symmetric graph has a $[6+6+6]$-type
($[4+4+10]$-type) embedding and the asymmetric graph has a $18$-type
embedding. Therefore we can say that the type of embedding is only
dependent on the ``degree of symmetry'' of the graph given by the
pair of values of the two sets $R$ and $S$. In other words, the
embeddings are topologically invariant under permutations acting on
the sets of even and odd vertices. Odd permutations on one set, merely
change the orientation of all vertices in the other set, in which
case the value of the set is not affected. Even permutations on one
set only affect the cyclic order of the orientations in that set,
e.g. if we have an orientation $(-1+1+1)$ of the set $R$ an even
permutation acting on $R$ would only result in, say, $(-1+1+1)\mapsto(+1+1-1)$. 

Thus, from all 64 possible configurations of the orientations of vertices
in $K_{3,3}$ only three of them induce inequivalent embeddings. If
the pair of values is $(3,3)$, then there are 4 equivalent configurations
which induce a $[6+6+6]$-type embedding; either all 6 vertices have
positive (or negative) orientation or 3 vertices from one set have
positive (or negative) orientation while the vertices from the other
set have opposite orientation. Therefore, we are left with $64-4=60$
configurations; $36$ from them belong to the case where the pair
of value is $(1,1)$. This is a $[4+4+10]$-type embedding, hence,
one of the vertices on each set has the opposite orientation relative
to the vertices from the set which belongs to, i.e. the sets have
the orientations of the form $(+1-1-1)$ and cyclic, or of the form
$(-1+1+1)$ and cyclic. Therefore, for each relative orientation there
are 3 cases, which make $3\times2=6$ for each set. The rest $24$
of the configurations belong to the case where the pair of values
is $(3,1)=(1,3)$. This gives an embedding in the double torus. There
are 6 cases where the value of a set is 1 and 2 cases where the value
is 3, hence, for each of the pairs $(1,3)$ and $(3,1)$ there are
12 configurations to consider.

Finally, we can summarize the above discussion by the following statement,
\begin{fact}
If the value of the two disjoint sets of vertices in $K_{3,3}$ is
unequal, then the cellular embedding of $K_{3,3}$ corresponds to
an embedding in $T^{2}\#T^{2}$; otherwise the only cellular embeddings
of $K_{3,3}$ are in the torus $T^{2}$, such that it is a $[6+6+6]$-type
embedding for the $(3,3)$-value or a $[4+4+10]$-type embedding for
the $(1,1)$-value.
\end{fact}
This result is important since it will allow us to extract information
of the terms needed in the evaluation of non-planar spin networks
to account for their topology%
\footnote{We mean by the topology of a graph, the topology of the surface in
which the graph is embedded.%
}. The graph contains topological information about the surface in
which it is cellular embeddable and we can use recoupling theory to
extract that information. The reason for this is that we are considering
only cellular embeddings and we use all the information contained
in the graph (number of edges, vertices and their orientation) to
build the surfaces by the Rotation Scheme. Hence, the information
of the topology of the surface must be contained in the graph itself;
in other words, by reducing the graph in the embedding, we receive
a factor in the evaluation that reflects the information of the graph
being non-planar. That is why it is important to consider only cellular
embeddings, the faces are only 2-cells homeomorphic to discs with
no information about the global topology. For instance, in the case
of the $K_{4}$ graph we have two cellular embeddings in the torus,
one with an embedded 3- and another with an embedded 4-cycle as cells
and both wrapping the two circles of the torus. Both cellular embeddings
give, in fact, different evaluations, however only up to a constant
involving powers of $q$ (or $A$). In fact, these spin networks are
contained in $K_{3,3}$, in the sense that reducing the graph of $K_{3,3}$
in the torus via Moussouris' algorithm leads to such diagrams. We
will call the embedding of the tetrahedron in the torus with an embedded
3-cycle the \textit{toroidal coupling coefficient}.

There are of course (non-cellular) embeddings of a graph in surfaces
with higher genus, however, it is not the graph containing the information
about the topology of the surface but the surface itself. If we consider
a non-cellular embedding in the torus of the complete graph on 4 vertices,
i.e. the tetrahedron, and we \textquotedblleft{}cut\textquotedblright{}
the surface along the edge of the graph wrapping the circle of the
torus, we will get a surface which is not homeomorphic to a disc and
which contains the information about the topology of the torus. Thus,
the graph in that configuration has no information about a non-trivial
topology.

Due to the classification of closed oriented surfaces we expect that
the information extracted from the torus is sufficient to extend Moussouris'
algorithm for the evaluation of planar spin networks to the non-planar
case. We believe that by knowing the evaluation of the spin network
corresponding to the torus we can use it to evaluate all spin networks
with higher genus in terms of products of this evaluation. To evaluate
these spin networks it would be necessary to arrange them such that
it is possible to \textquotedblleft{}cut\textquotedblright{} their
components (using the generalized Wigner-Eckart theorem, cf. \cite{moussouris1983quantum})
corresponding to each handle of the oriented surface and evaluate
each torus separately, this would give a sum of products of toroidal
symbols%
\footnote{cf. Def. \vref{The-toroidal-symbol}.%
}.

\section{The Evaluation of Non-planar Spin Networks\label{sec:non-planar-SN}}

In this section we discuss Moussouris' Decomposition Theorem and present
its algorithm for the evaluation of planar spin networks, which relates
these to the Ponzano-Regge partition function. We then give an improved
version of the theorem and apply this algorithm to the graph $K_{3,3}$
in order to extract the information needed to extend it to non-planar
spin networks, i.e. we give the explicit form of the toroidal phase
factor for the $q$-deformed case and discuss what needs to be done
to achieve a generalization of the mentioned theorem.

\subsection{The Decomposition Theorem\label{sub:The-Decomposition-Theorem}}

In \cite{moussouris1983quantum} J. P. Moussouris proved a theorem
which relates the spin networks with the Ponzano-Regge theory by reducing
a recoupling graph%
\footnote{Recall that a recoupling graph of a group $G$ is a labelled 3-valent
graph representing a contraction of tensors of $G$. In the following,
the term ``recoupling graph'' will denote such a graph together
with an orientation of its vertices.%
} to a sum of products of coupling coefficients. This reduction, known
as the Decomposition Theorem, gives an evaluation of the spin network
only dependent on the labeling of the graph, as in \cite{turaev1992state}
for a manifold with boundary.

First, we give the explicit form of the identities used in Moussouris'
algorithm in the form of the following definition.
\begin{defn}
\textit{\label{Moussouris-algorithm}Moussouris' algorithm} is defined
to be the application of a finite number of the following operations,
\begin{itemize}
\item The crossing identity:
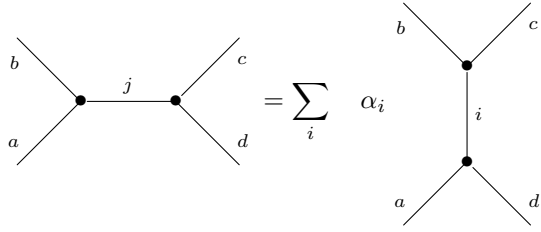
\[
\xygraph{!{0;/r2pc/:}[d]!{\sbendh-@(0)>{a}}[ul]!{\sbendv@(0)<{b}}*{\bullet}!{\xcaph[1.5]@(0)|{j}}[r(0.5)]*{\bullet}!{\sbendh@(0)>{c}}[dl]!{\sbendv-@(0)<{d}}}=\sum_{i}\hspace{1em}\alpha_{i}\xygraph{!{0;/r2pc/:}[l][u(1.55)]!{\sbendv@(0)<{b}}*{\bullet}!{\xcapv[1.5]@(0)|{i}}[u]!{\sbendh@(0)>{c}}[d(3.5)][l(2)]!{\sbendh-@(0)>{a}}*{\bullet}!{\sbendv-@(0)<{d}}}
\]
where $\alpha_{i}=\Delta_{i}\Bigl\{\begin{array}{ccc}
a & b & i\\
c & d & j
\end{array}\Bigr\}_{q}$; $\Delta_{i}=(-1)^{2i}[2i+1]$ and $[*]$ denotes the quantum integer
defined as $\left[n\right]=\frac{q^{n}-q^{-n}}{q-q^{-1}}$.
\item The excision identity:
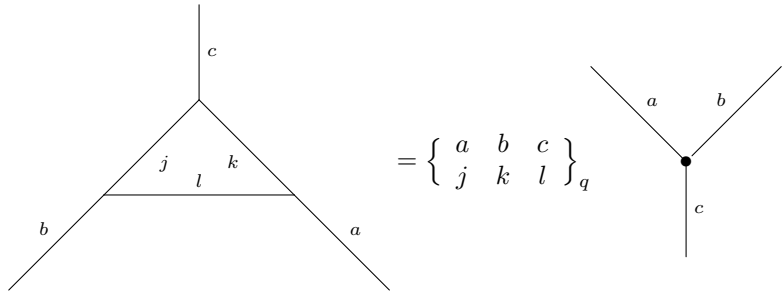
\[
\xygraph{[d(0.35)]!{\sbendh@(0)|{j}}!{\sbendv@(0)|{k}}[ll]!{\xcaph[2]@(0)|{l}}[uu]!{\xcapv@(0)|{c}}[lldd]!{\sbendh-@(0)|{b}}[rr]!{\sbendv@(0)|{^{a}}}}=\Bigl\{\begin{array}{ccc}
a & b & c\\
j & k & l
\end{array}\Bigr\}_{q}\xygraph{*{\bullet}!{\sbendh-@(0)|{b}}[ld]!{\xcapv@(0)|{c}}[uul]!{\sbendv@(0)|{^{a}}}}
\]
 
\end{itemize}
\end{defn}
There are two versions of the mentioned theorem which we will present
and analyze in this section in order to understand how the expansion
of the algorithm for evaluating non-planar spin networks could be
done. The second version of the theorem, called network version, is
more general than the first one since it does not assume the spin
network to be planar, however, it assumes implicitly the existence
of an embedded cycle for the recoupling graph $F$ to be reduced.
This implies that the embedding of the graph in some surface has at
least two 2-cells since a cycle induces a region homeomorphic to a
disc by using only one side of each edge in the Rotation Scheme. First,
we present without proof the planar version of the theorem found in
\cite{moussouris1983quantum} and after a short discussion we arrive
to the heart of this paper and give the improved network version and
its proof.
\begin{thm}
\textbf{\textup{Planar version of the Decomposition Theorem:}}

Let $F$ be a \textbf{planar }recoupling graph and $D(F)$ its dual
relative to a particular embedding in the sphere. Let $C(F)$ be a
combinatorial 3-manifold produced by dissecting $D(F)$ with internal
edges $x_{1},x_{2},\dots,x_{p}$ into tetrahedra $T_{1},\dots,T_{q}$.
Then, the evaluation of the recoupling graph is given by the amplitude
\[
\Psi(F)=\sum_{x_{1},\dots,x_{p}}\prod_{j=1}^{p}[x_{j}]\prod_{k=1}^{q}[T_{k}]
\]
where $[x_{j}]$ is the loop-value of the edge $x_{j}$ and the $[T_{k}]$'s
are the coupling coefficient associated with the tetrahedra $T_{k}$.
\end{thm}
The successive application of the Alexander moves, which correspond
in the dual form to the elimination of a 3-cycle and the crossing
identity, results in the introduction of sufficient internal edges
to dissect the interior of $D(F)$ into tetrahedra, giving a combinatorial
3-manifold $C(F)$ with $D(F)$ as its boundary. This decomposition
process is non-unique, however, the Biedenharn-Elliott identity and
the orthogonality of the $6j$-symbols ensures the equivalence of
the decompositions, \cite{moussouris1983quantum}. This is a special
case of the procedure to obtain the invariant%
\footnote{Notice that in the amplitude given above the theta-net factors are
missing. This is due to the fact that in \cite{moussouris1983quantum}
the spin networks are normalized such that the theta-nets are evaluated
to one. %
} described in \cite{turaev1992state}.

Moussouris does mention the importance of the orientation of the vertices
in the evaluation of the graph, pointing out that considering the
orientation of the vertices results in so-called phase factors, which
can be isolated as values of graphs with two vertices with the same
orientation. However, in the proof of the network version of the theorem
this consideration enters only in the first and second steps of the
induction on the number of vertices $V$ in $F$, disregarding the
fact that the orientation of the vertices of a recoupling graph affects
the embedding of it in a surface, which might be such that there is
no embedded cycle at all. This would mean that the spin network cannot
be reduced straightforward. We will discuss this case later. Moreover,
in the proof it is also assumed implicitly that after reducing all
embedded cycles the only diagram left is either a coupling coefficient
or a toroidal phase factor (cf. \subref{The-Toroidal-SN}); in the
latter case we can call such a recoupling graph a toroidal spin network
since the phase factor left at the end of the reduction contains the
information of the graph being embedded in the torus%
\footnote{The coupling coefficient with a toroidal phase factor has as its cellular
embedding exactly the one discussed at the end of the previous section.%
}. 

As seen in example \exaref{case-3-1}, the appearance of an embedded
cycle is not always the case and there exist spin networks which are
irreducible if we only consider the operations described in \cite{moussouris1983quantum},
thus, Moussouris' Decomposition Theorem is limited to planar and toroidal
spin networks. Hence, it is necessary to rewrite the Decomposition
Theorem in a more precise manner in order to consider the case where
the spin network is at most toroidal, i.e. with toroidal phase factors,
which are especially important for the $q$-deformed case.

We will now give an improved version of the theorem and its proof
following \cite{moussouris1983quantum}. We modified the network version
of the theorem stated in Moussouris' Ph.D. thesis to account for the
discussion above.
\begin{thm}
\textbf{\textup{\label{thm:Decomposition-Theorem}Decomposition Theorem: }}

A recoupling graph $F$ of a compact semi-simple group $G$, which
is at most toroidal, can always be evaluated as a sum of products
of coupling coefficients of $G$ and toroidal phase factors.\end{thm}
\begin{proof}
First we discuss the existence of an embedded cycle in any (at most)
toroidal spin network from a graph-theoretical perspective. The fact
that the spin network is at most toroidal means that it is either
planar or cellular embeddable in the torus, in the former case the
number of faces is at least 3 because a 3-valent graph has the number
of edges proportional to the number of vertices 
\[
e=\frac{3}{2}v
\]
Thus, $v-e+f=2\Rightarrow2f-4=v$ and we need at least $v=2$ to have
a closed spin network. The existence of more than one 2-cell in a
planar graph implies automatically the existence of an embedded cycle
because having an embedded circuit means that somewhere along the
walk on the circuit one has to ``change'' the side of the walk to
return to the edge repeated in the opposite direction. On a plane
and for closed graphs this can only be done by a loop in the walk,
which bounds a region on the plane, i.e. it contains at least one
embedded cycle. In this case the graph corresponding to the spin network
contains a bridge, thus is not (vertex) connected. However, as we
discussed above, we can be sure of the existence of an embedded cycle,
even in this special case. For the more general case of higher connected
graphs we have in fact from Steinitz's theorem, stating that every
3-connected%
\footnote{A k-connected graph is one which stays connected even if one removes
any number of vertices smaller than k.%
} planar graph corresponds to a convex polyhedron, that planar 3-connected
spin networks contain only cycles.

On the other hand, if the graph is cellular embeddable in the torus
then $v-e+f=0\Rightarrow2f=v$ thus if $v\geq4$ there are at least
2 faces. This is the case for a toroidal coupling coefficient, where
the embedding of $K_{4}$ is non-planar. It contains an embedded 3-cycle
and an embedded 9-circuit, thus the circuits should be related to
the global topology of the surface.%
\footnote{The existence of a circuit for 3-connected graphs may imply non-planarity
in general, but not the other way around, since the (6+6+6)-type embedding
contains no circuit. We only consider this remark as an interesting
point without pretending to state a general fact.%
} Now, consider the case where we have one 2-cell, then the graph represents
a toroidal phase factor. To construct bigger 3-valent graphs cellular
embeddable in the torus one has to introduce vertices and edges in
such a way, that the resulting graph has an even number of vertices
and embeddable cycles or circuits. The existence of an embedded cycle
is already in the case of $v=4$ with the help of Moussouris' algorithm
assured. Introducing two more vertices and an edge results in three
possible types of graphs; two of which contain an embedded cycle and
one that has two circuits. However, the latter one can be transformed
(in the sense of spin networks) into a graph containing an embedded
cycle%
\footnote{This is a very special case, which could also be considered as being
two toroidal phase factors connected imposing a very strong condition
for the spin network not to vanish.%
}. Following these thoughts one concludes the existence of embedded
cycles in more complex toroidal spin networks.

Having discussed the existence of embedded cycles in at most toroidal
spin networks we turn now to the proof of the decomposition via Moussouris'
algorithm. The proof is by induction on the number of vertices $V$
in $F$ and the size $l$ of the smallest embedded cycle.

If $V=2$, the recoupling graph is a toroidal phase factor or a theta-evaluation
of a vertex. The case $V=3$ is not possible.

If $V=4$, the recoupling graph is a (toroidal) coupling coefficient
or two phases.

If $V>4$, we look for the smallest embedded cycle in $F$ and reduce
it as follows, depending on its size $l$. A 2-cycle is reduced using
Schur's identity. This results in a new graph containing $V-2$ vertices.
A 3-cycle is eliminated by producing a single coupling coefficient
by the Wigner-Eckart theorem. Alternatively, one can regard the so-called
``crossing identity'' described below to reduce the 3-cycle to a
2-cycle and apply Schur's identity. The resulting graph contains $V-2$
vertices.

For the case $l>3$ we have an embedded cycle with $l$ edges labeled
by $j_{1},j_{2},\dots,j_{l}$. This reduces to a $(l-1)$-cycle by
the crossing identity derived from using the Recoupling Theorem on
the edge, say $j_{l}$, as depicted below. This operation results
in a coupling coefficient multiplied by a recoupling graph in which
the edge $j_{l}$ is removed while a new ``internal'' edge $x$
is introduced, coupling $j_{1}$ to $j_{l-1}$ and the other two edges,
which were coupling to $j_{l}$, are also coupled to $x$ and to each
other. This resulting product is summed over the new edge $x$ as
in the Recoupling Theorem. The cycle is then reduced until $l=3$.

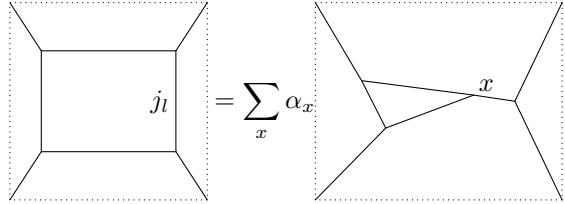
\[
\xy/r2pc/:{\xypolygon4"S"{~:{(1.5,0):(0,.75)::}}},{\xypolygon4"L"{~:{(2.2,0):(0,1)::}~>{.}}},"S1"-(0.25,0.8)*{j_{l}},"L1";"S1"**@{-},"L2";"S2"**@{-},"S4";"L4"**@{-},"S3";"L3"**@{-}\endxy=\sum_{x}\alpha_{x}\xy/r2pc/:{\xypolygon4"L"{~:{(2.75,0):(0,.8)::}~>{.}}},-(.5,0)="B",{\xypolygon3"S"{~:{(.8,0.3):(-0.65,.65)::}}},"L3";"S2"**@{-},"L2";"S1"**@{-},"S3";"L0"+(1.2,0)="K"**@{-},"L1";"K"**@{-},"L4";"K"**@{-},"L0"+(.75,.25)*{x}\endxy
\]
 where $\alpha_{x}=\Delta_{x}\Bigl\{\begin{array}{ccc}
a & b & x\\
c & d & j
\end{array}\Bigr\}$ are the so called coupling coefficients, a product of a $6j$-symbol
and the ``value'' of the edge $x$ give by $\Delta_{x}=(-1)^{2x}[2x+1]$,
with the brackets meaning that the integer $2x+1$ is taken to be
a quantum integer for $q\neq1$.

This process of vertex reduction is repeated until $V=4$ giving as
a result a product of coupling coefficients and a phase factor summed
over all internal variables.
\end{proof}

\subsection{\label{sub:The-Toroidal-SN}The Evaluation of the Toroidal Spin Network
$K_{3,3}$ }

We will now apply the algorithm described in the above proof to specific
spin networks of the $[4+4+10]$- and $[6+6+6]$-type embeddings of
$K_{3,3}$ on a torus, denoted by $K_{3,3}^{(-1,+1)}$ and $K_{3,3}^{(-3,+3)}$
respectively%
\footnote{Observe that the labels $(-1,+1)$ and $(-3,+3)$ have the sign of
the overall sum over the orientations of each set and not only the
value of the set. The reason for this is that the evaluation depends
on the configuration of the orientations on each set of $K_{3,3}$,
as we will see later.%
}. This will be done in order to extract the information about the
topology encoded in the graph for the evaluation of non-planar spin
networks. 

In the case of the embedding $K_{3,3}^{(-1,+1)}$ we may start by
applying the crossing identity to the common edge of the 4-cycles
and then eliminating the two resulting 3-cycles by extracting two
coupling coefficients.

\bigskip{}

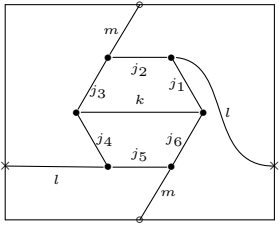
\begin{center}
\def\objectstyle{\scriptscriptstyle}
\xy/r2pc/:
{\xypolygon4"T"{~:{(3,0):(0,.8)::}}},{\xypolygon6"A"{~:{(1,0):(0,1)::}{\bullet}}},
"T1";"T2"**{}?<>(.5)="T5"*{\circ},"T4";"T3"**{}?<>(.5)="T6"*{\circ},
"T1";"T4"**{}?<>(.5)="T7","T2";"T3"**{}?<>(.5)="T8",
"A1";"A4"**@{-},"A3";"T5"**@{-},"A6";"T6"**@{-},
"T7";"T4"**{}?<>(.5)="T9"*{\times},"T8";"T3"**{}?<>(.5)="T10"*{\times},
"T10";"A5"**@{-},"T9";"A2"**\crv{"A6"+(.5,0)&"A2"+(1,0)},
"A1"-(.15,-.15);"A2"-(.15,.15)**{}?<>(.5)*{j_1},"A2"-(0,.2);"A3"-(0,.2)**{}?<>(.5)*{j_2},"A3"-(0,.2);"A4"+(.2,0)**{}?<>(.5)*{j_3},
"A4"+(.2,0);"A5"+(.2,0)**{}?<>(.5)*{j_4},"A5"+(0,.2);"A6"+(0,.2)**{}?<>(.5)*{j_5},"A6"-(.2,0);"A1"-(.2,0)**{}?<>(.5)*{j_6},
"A1"+(0,.2);"A4"+(0,.2)**{}?<>(.5)*{k},
"A3"-(.2,0);"T5"-(.2,0)**{}?<>(.5)*{m},"A6"+(.2,0);"T6"+(.2,0)**{}?<>(.5)*{m},
"A5"-(0,.2);"T10"-(0,.2)**{}?<>(.5)*{l},"A1";"T7"**{}?<>(.3)*{l},
\endxy
\par\end{center}

\bigskip{}

The result is a sum over a single internal edge $x$ of a product
of three $6j$-symbols weighted by a factor of $(-1)^{2x}[2x+1]$.
These are, however, not all the factors since the diagram left encodes
the information of the graph being embedded in a torus. This diagram,
which we will call \textbf{toroidal phase factor}, can be represented
in a torus as follows:\smallskip{}

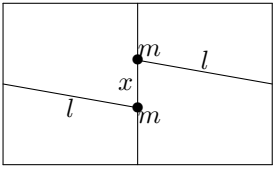
\begin{center}
\xy/r1.5pc/:{\xypolygon4"T"{~:{(4,0):(0,.6)::}}},
"T0"-(.25,0)*{x},"T0"-(0,.5)="X1"*{\bullet},"T0"+(0,.5)="X2"*{\bullet},
"X1";"X2"**@{-},"T1";"T2"**{}?<>(.5)="M2";"X2"**@{-},"T4";"T3"**{}?<>(.5)="M1";"X1"**@{-},
"X1"+(.25,-.2)*{m},"X2"+(.25,.2)*{m},
"T1";"T4"**{}?<>(.5)="L2";"X2"**@{-},"T3";"T2"**{}?<>(.5)="L1";"X1"**@{-},
"L1";"X1"**{}?<>(.5)-(0,.25)*{l},"L2";"X2"**{}?<>(.5)+(0,.25)*{l}
\endxy
\par\end{center}

\smallskip{}

If we ``project'' this diagram to the plane by connecting the loose
ends of the edges $m$ and $l$, once we have disregarded the frame
of the above diagram, we get a theta-net with these edges crossing.
In order to get a more familiar theta-net, which can then be set to
have the value of 1, we need to ``twist'' the edge $x$. This is
done by following operation on a vertex%
\footnote{\label{fn:definition-of-twist-factor-under-cross}Notice that here
we are presenting the case of a vertex with an over-crossing, however,
this operation is defined for an under-crossing as well. In this case
we exchange $A\rightarrow A^{-1}$. This amounts to the choice of
orientation in the corresponding surface. Furthermore, observe that
there is an arbitrary choice of the twist factor since one can change
the sign in the exponent by performing some isotopies and using the
fact that a curl in an edge $a$ gives a factor of $(-1)^{-2a}A^{\pm4a(a+1)}$,
\cite{kauffman1994temperley}.%
}, \cite{carter1995classical,kauffman1994temperley},
\begin{equation}
\xygraph{!{0;/r1pc/:}[u(2)]!{\xcapv@(0)>{j}}*{\bullet}[lr(0.1)]!{\sbendv-@(0)|{b}}[ll]!{\sbendh-@(0)|{a}}[ld]!{\xunderv[2]}}=(-1)^{a+b-j}A^{2[a(a+1)+b(b+1)-j(j+1)]}\xygraph{!{0;/r1.5pc/:}[u]!{\xcapv@(0)>{j}}*{\bullet}[lr(0.1)]!{\sbendv-@(0)|{a}}[ll]!{\sbendh-@(0)|{b}}},\label{eq:Twist-factor}
\end{equation}
where $A=q^{2}$ is the deformation parameter of the underlying quantum
group. The result of applying Moussouris algorithm and the above twisting
rule is a sum of products of coupling coefficients as in the planar
case, however, the non-planar nature of the graph is reflected in
the ``twist factor'' given above, i.e. in the evaluation of the
toroidal phase factor. Thus, we have 
\begin{equation}
\big[K_{3,3}^{(-1,+1)}\big]=\sum_{x}\Delta_{x}\left\{ \begin{array}{ccc}
j_{3} & j_{4} & k\\
j_{6} & j_{1} & x
\end{array}\right\} \left\{ \begin{array}{ccc}
m & j_{5} & j_{6}\\
j_{4} & x & l
\end{array}\right\} \left\{ \begin{array}{ccc}
j_{2} & l & j_{1}\\
x & j_{3} & m
\end{array}\right\} (-1)^{l+m-x}A^{2[l(l+1)+m(m+1)-x(x+1)]}\label{eq:quantum-9j-symbol}
\end{equation}
where $\Delta_{x}=(-1)^{2x}[2x+1]_{q}$ is the loop-evaluation. The
squared brackets on the left hand side of the equation denote the
evaluation of a graph in terms of coupling coefficients.
\begin{rem*}
The \eqref{quantum-9j-symbol} looks similar to the $9j$-symbol.
However, considering the cases where $A=\pm1$, we have an overall
factor of $(-1)^{l+m+x}$ which corresponds to one of the coupling
coefficients having a vertex with the ``wrong'' orientation. This
can be seen by expressing one of the $6j$-symbols where the labels
$x,l,m$ form an admissible triple in terms of $3j$-symbols and permuting
the order of the labels in the corresponding $3j$-symbol by the following
relation, \cite{edmonds1996angular},
\[
(-1)^{j_{1}+j_{2}+j_{3}}\left(\begin{array}{ccc}
j_{1} & j_{2} & j_{3}\\
m_{1} & m_{2} & m_{3}
\end{array}\right)=\left(\begin{array}{ccc}
j_{2} & j_{1} & j_{3}\\
m_{2} & m_{1} & m_{3}
\end{array}\right).
\]
The resulting factor is exactly the one described in \cite[p. 65]{moussouris1983quantum},
i.e. a tetrahedron with one of the vertices having an orientation
so that two edges cross. In fact, the diagram left after applying
the crossing identity and eliminating only one of the two 3-cycle
gives such a tetrahedron. Notice that if we disregard the information
encoded in the orientation of the vertices, i.e. we set the phase
factor equal to 1, we get the usual $9j$-symbol.\end{rem*}
\begin{defn}
\label{The-toroidal-symbol}The \textit{toroidal symbol }is defined
in the following way
\[
\left[\begin{array}{ccc}
j_{1} & j_{4} & j_{7}\\
j_{2} & j_{5} & j_{8}\\
j_{3} & j_{6} & j_{9}
\end{array}\right]_{A,+}^{(j_{2},j_{6})}:=\sum_{x}\Delta_{x}\mathcal{A}_{j_{2},j_{6},x}^{+}\left\{ \begin{array}{ccc}
j_{1} & j_{2} & j_{3}\\
j_{6} & j_{9} & x
\end{array}\right\} \left\{ \begin{array}{ccc}
j_{4} & j_{5} & j_{6}\\
j_{2} & x & j_{8}
\end{array}\right\} \left\{ \begin{array}{ccc}
j_{7} & j_{8} & j_{9}\\
x & j_{1} & j_{4}
\end{array}\right\} 
\]
where $\mathcal{A}_{j_{2},j_{6},x}^{+}=(-1)^{j_{2}+j_{6}-x}A^{+2[j_{2}(j_{2}+1)+j_{6}(j_{6}+1)-x(x+1)]}$
and we call the labels $(j_{2},j_{6})$ the indices of the toroidal
symbol.

The symbol corresponding to the evaluation (\ref{eq:quantum-9j-symbol})
is given by
\begin{equation}
[K_{3,3}^{(-1,+1)}]=\left[\begin{array}{ccc}
j_{1} & j_{6} & k\\
l & j_{5} & j_{4}\\
j_{2} & m & j_{3}
\end{array}\right]_{A,+}^{(l,m)}.\label{eq:Toroidal-symbol-for-(-1,+1)}
\end{equation}

\end{defn}
Consider now the evaluation $\big[K_{3,3}^{(-3,+3)}\big]$. It is
possible to reduce the $[6+6+6]$-type embedding to the $[4+4+10]$-type
one by applying the crossing identity on two edges of the hexagonal
figure shown in example \exaref{case-3-3} which belong to the same
``exterior'' 6-cycle, for instance the edges $j_{2}=(23)$ and $j_{6}=(14)$

\bigskip{}

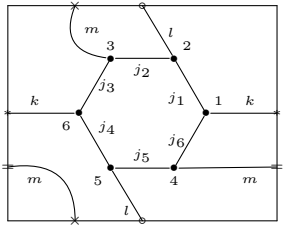
\begin{center}
\def\objectstyle{\scriptscriptstyle}
\xy/r2pc/:
{\xypolygon4"T"{~:{(3,0):(0,.8)::}}},{\xypolygon6"A"{~:{(1,0):(0,1)::}{\bullet}}},
"T1";"T2"**{}?<>(.5)="T5"*{\circ},"T4";"T3"**{}?<>(.5)="T6"*{\circ},
"T1";"T4"**{}?<>(.5)="T7"*{*},"T2";"T3"**{}?<>(.5)="T8"*{*},
"A1";"T7"**@{-},"A4";"T8"**@{-},"A2";"T5"**@{-},"A5";"T6"**@{-},
"T7";"T4"**{}?<>(.5)="T9"*{=},"T8";"T3"**{}?<>(.5)="T10"*{=},
"T5";"T2"**{}?<>(.5)="T11"*{\times},"T6";"T3"**{}?<>(.5)="T12"*{\times},
"T9";"A6"**@{-},"T10";"T12"**\crv{"T12"+(0,1)},"T11";"A3"**\crv{"T11"-(.3,.75)},
"A1"+(.2,.2)*{1},"A2"+(.2,.2)*{2},"A3"+(0,.2)*{3},"A4"-(.2,.2)*{6},"A5"-(.2,.2)*{5},"A6"-(0,.2)*{4},
"A1"-(.2,.2);"A2"-(.2,.2)**{}?<>(.5)*{j_1},"A2"-(.2,.2);"A3"-(-.2,.2)**{}?<>(.5)*{j_2},"A4"+(.2,.2);"A3"-(-.2,.2)**{}?<>(.5)*{j_3},
"A4"+(.2,.2);"A5"+(.2,.2)**{}?<>(.5)*{j_4},"A5"+(.2,.2);"A6"+(-.2,.2)**{}?<>(.5)*{j_5},"A6"-(.2,-.2);"A1"-(.2,.2)**{}?<>(.5)*{j_6},
"T7"+(0,.2);"T8"+(0,.2)**{}?<>(.1)*{k},"T7"+(0,.2);"T8"+(0,.2)**{}?<>(.9)*{k},
"T9"-(0,.2);"T10"-(0,.2)**{}?<>(.1)*{m},"T9"-(0,.2);"T10"-(0,.2)**{}?<>(.9)*{m},"T11";"A3"**{}?<>(.5)*{m},
"T5"+(.2,0);"A2"+(.2,0)**{}?<>(.5)*{l},"T6"-(0,.25);"A5"-(0,.25)**{}?<>(.5)*{l},
\endxy
\par\end{center}

\bigskip{}

From this procedure we get
\begin{equation}
\big[K_{3,3}^{(-3,+3)}\big]=\sum_{x,y}\Delta_{x}\Delta_{y}\left\{ \begin{array}{ccc}
j_{1} & j_{2} & l\\
m & x & j_{3}
\end{array}\right\} \left\{ \begin{array}{ccc}
j_{1} & j_{6} & k\\
m & y & j_{5}
\end{array}\right\} \left[\begin{array}{ccc}
j_{5} & l & j_{4}\\
y & m & k\\
j_{1} & x & j_{3}
\end{array}\right]_{A,+}^{(x,y)}\label{eq:Reduction-(3,3)-to-(1,1)}
\end{equation}
where the first $6j$-symbol is the result of the crossing identity
on the edge $j_{2}=(23)$ and the second one is the result of the
same identity on the edge $j_{6}=(14)$.

When reducing the spin network $K_{3,3}^{(-3,+3)}$ by applying the
algorithm on edges of different ``exterior'' $6$-cycles, for instance
on $j_{2}$ and $j_{5}$, the reduction takes longer, however, again
the Biedenharn-Elliott identity ensures that the evaluation is independent
of the decomposition procedure.

At this point it is important to describe some properties of the toroidal
symbol defined above. The symmetries of the (quantum) $6j$-symbols
make the toroidal symbol invariant under reflection on either of the
diagonals. Furthermore, if one exchanges two columns and two rows
leaving one of the arguments on the anti-diagonal fixed, then the
toroidal symbol stays invariant.%
\footnote{Notice that the arguments in the anti-diagonal are the ones paired
with the argument of the three $6j$-symbols over which it is summed,
i.e. they only appear once in the product of $6j$-symbols.%
} Moreover, there is an important new property that makes it possible
to change the indices of the toroidal symbol and to exchange the arguments
in the anti-diagonal. The following identity is due to a relation
between sums over products of three quantum $6j$-symbols found in
\cite{carter1995classical}%
\footnote{The relation is given in the reference in a different form, namely,
as a sum of products of three $6j$-symbols and a factor similar to
the twist factor described above. We used the proof for the case $A=\pm1$
given in \cite{carter1995classical} as a guide to reconstruct the
relation in order to present it as a ``(pseudo-)symmetry'' of the
toroidal symbol. %
},
\begin{equation}
\left[\begin{array}{ccc}
y & k & m\\
j_{5} & j_{4} & l\\
j_{1} & j_{3} & x
\end{array}\right]_{A,+}^{(x,y)}=\mathcal{A}_{k,j_{5},(j_{1}+m)}^{+}\left[\begin{array}{ccc}
y & j_{1} & j_{5}\\
m & x & l\\
k & j_{3} & j_{4}
\end{array}\right]_{A,+}^{(y,j_{4})}.\label{eq:exchange identity}
\end{equation}
Observe that the notation $\mathcal{A}_{k,j_{5},(j_{1}+m)}^{+}$ above
does not mean that the exponent of $A$ corresponding to the third
index is $-(j_{1}+m)((j_{1}+m)+1)$, but rather $-(j_{1}(j_{1}+1)+m(m+1))$.
Notice also that using this relation twice amounts to the identity,
thus, this transformation can be regarded as a ``(pseudo-)symmetry''.

If we compare $\big[K_{3,3}^{(-1,+1)}\big]$ with the\textit{ }$9j$-symbol
we may recognize the possibility to use the following relation between
a $9j$-symbol (without twist factor) and $6j$-symbols, as in \cite{edmonds1996angular},
\begin{equation}
\sum_{\mu}(2\mu+1)\left\{ \begin{array}{ccc}
j_{11} & j_{12} & \mu\\
j_{21} & j_{22} & j_{23}\\
j_{31} & j_{32} & j_{33}
\end{array}\right\} \left\{ \begin{array}{ccc}
j_{11} & j_{12} & \mu\\
j_{23} & j_{33} & \lambda
\end{array}\right\} =(-1)^{2\lambda}\left\{ \begin{array}{ccc}
j_{21} & j_{22} & j_{23}\\
j_{12} & \lambda & j_{32}
\end{array}\right\} \left\{ \begin{array}{ccc}
j_{31} & j_{32} & j_{33}\\
\lambda & j_{11} & j_{21}
\end{array}\right\} \label{eq:9j-and-6j-relation}
\end{equation}
However, this relation does not account for the twist factor, hence,
it is not possible to use straightforward.
\begin{claim}
The identity corresponding to \eqref{9j-and-6j-relation} is
\begin{equation}
\sum_{\mu}[2\mu+1]\left[\begin{array}{ccc}
j_{1} & j_{4} & \mu\\
j_{2} & j_{5} & j_{8}\\
j_{3} & j_{6} & j_{9}
\end{array}\right]_{A,+}^{(j_{2},j_{6})}\left\{ \begin{array}{ccc}
j_{1} & j_{4} & \mu\\
j_{8} & j_{9} & \lambda
\end{array}\right\} =\mathcal{A}_{j_{2},j_{6},\lambda}^{+}\left\{ \begin{array}{ccc}
j_{2} & j_{5} & j_{8}\\
j_{4} & \lambda & j_{6}
\end{array}\right\} \left\{ \begin{array}{ccc}
j_{3} & j_{6} & j_{9}\\
\lambda & j_{1} & j_{2}
\end{array}\right\} \label{eq:Relation-to-reduce-(3,3)-to-(1,1)}
\end{equation}
where $[2\mu+1]$ is a quantum integer corresponding to the loop-value
of $\mu$.\end{claim}
\begin{proof}
The only term containing $\mu$ in the expansion of the l.h.s. in
term of $6j$-symbols is of the form
\[
\sum_{\mu}[2x+1][2\mu+1]\left\{ \begin{array}{ccc}
j_{11} & j_{12} & \mu\\
j_{23} & j_{33} & x
\end{array}\right\} \left\{ \begin{array}{ccc}
j_{11} & j_{12} & \mu\\
j_{23} & j_{33} & \lambda
\end{array}\right\} =\delta_{x,\lambda}.
\]
Thus, the only term left is the one on the r.h.s.\end{proof}
\begin{rem*}
Notice that (\ref{eq:Relation-to-reduce-(3,3)-to-(1,1)}) only holds
if the toroidal symbol has that exact form, i.e. $\mu$ must be in
any counter-diagonal position. Labels in that position appear in the
expansion \eqref{quantum-9j-symbol} only in one $6j$-symbol, thus,
the orthogonality of the $6j$-symbols may be used straightforward.
Moreover, it is possible to transpose the toroidal symbol since this
only changes the ordering of the admissible triples in the $6j$-symbols,
i.e. it is possible to use the symmetry properties of the $6j$-symbols
to achieve a transposition of the toroidal symbol. In the regular
case without twist factor, the $9j$-symbols have some symmetries
and this constraint does not appear. However, one can use the identity
\vref{eq:exchange identity}, hence, even without having the symmetries
needed it is possible to transform the toroidal symbol in (\ref{eq:Reduction-(3,3)-to-(1,1)})
in order to bring it in a form suitable for the use of (\ref{eq:Relation-to-reduce-(3,3)-to-(1,1)})
to achieve a further simplification of $\big[K_{3,3}^{(-3,+3)}\big]$.
Thus, assuming the triple $(y,j_{4},j_{2})$ is admissible and using
the following relation 
\[
\sum_{x}[2x+1]\left\{ \begin{array}{ccc}
j_{1} & j_{2} & l\\
m & x & j_{3}
\end{array}\right\} \left[\begin{array}{ccc}
y & j_{1} & j_{5}\\
m & x & l\\
k & j_{3} & j_{4}
\end{array}\right]_{A,+}^{(y,j_{4})}=\mathcal{A}_{y,j_{4},j_{2}}^{+}\left\{ \begin{array}{ccc}
y & m & k\\
j_{3} & j_{4} & j_{2}
\end{array}\right\} \left\{ \begin{array}{ccc}
j_{2} & l & j_{1}\\
j_{5} & y & j_{4}
\end{array}\right\} 
\]
we can simplify (\ref{eq:Reduction-(3,3)-to-(1,1)}) further.
\end{rem*}
Summarizing the discussion above we obtain the result that the $[6+6+6]$-type
embedding of the $(3,3)$-bipartite graph as a spin network has following
evaluation
\begin{equation}
\big[K_{3,3}^{(-3,+3)}\big]=\mathcal{A}_{2k,(j_{5}+j_{4}),(j_{1}+j_{2})}^{+}\left[\begin{array}{ccc}
j_{1} & k & j_{6}\\
j_{2} & j_{3} & m\\
l & j_{4} & j_{5}
\end{array}\right]_{(-)}^{(k,m)}.\label{eq:9j-symbol?}
\end{equation}
Notice the change of sing of the toroidal symbol. This is due to the
choice of the orientation. In other words, observe that the edges
$l$ and $m$ wrap the torus in both examples above in fact in a different
way, thus, this amounts to choosing different crossings and with it
different signs in $\mathcal{A}_{.,.,.}^{\pm}$. Furthermore, observe
that the idea of using the twisting rule to change the orientation
of a vertex and with it its embedding before evaluating the given
spin network is highly arbitrary and naive. The factor resulting out
of this process is not the same as the factor gotten above as a result
of evaluating the spin network under consideration via the extended
algorithm and the identities above. 

It is interesting to observe the fact that the proportionality constant
reflects the possible maximal values that the spin label $l$ could
have if seen as a summand of the decomposition of the tensor product
of $j_{5}$ and $j_{4}$ or $j_{1}$ and $j_{2}$ giving the admissibility
conditions in the respective vertex, where the intertwiners $j_{5}\otimes j_{4}\rightarrow l$
and $j_{1}\otimes j_{2}\rightarrow l$ sit. 

Notice that the above relation is not exactly the toroidal symbol
in (\ref{eq:Toroidal-symbol-for-(-1,+1)}), however, it is proportional
to some toroidal symbol up to a choice of orientation of the surface
in which it is embedded, cf. \fnref{definition-of-twist-factor-under-cross}.
From this we learn that the orientation of the vertices change the
evaluation of the spin network. To see exactly how this happens further
work is needed. For instance, how does any two spin networks with
the same embedding relate to each other? Are they related via the
transformation (\ref{eq:exchange identity}) above or some symmetries?
However, the general form of the symbol remains and we can observe
that the topology of the surface in which the diagram is embedded
is reflected in the evaluation of the spin network. Nevertheless,
it is important to observe that even considering a single type of
embedding, say $[4+4+10]$, e.g. with a fixed value $\mathfrak{v}$
of the set of vertices, the toroidal symbols are not all equal for
different orientations but same $\mathfrak{v}$, although the difference
might be only a proportionality factor changing the indices of the
symbol, for instance,
\[
[K_{3,3}^{(-1,-1)}]^{(j_{4},j_{6})}=\mathcal{A}_{j_{4},j_{6},(l+m)}^{+}[K_{3,3}^{(+1,+1)}]^{(l,m)},
\]
hence, to find out how many distinct toroidal symbols actually are,
and the relation among them, some further work is needed in understanding
the relation between different orientations of the vertices with the
same resulting embedding.

Finally, we consider the embedding of $K_{3,3}$ in the double torus.
As mentioned before, this embedding has only one 2-cell, thus, is
not possible to reduce it by means of Moussouris' algorithm. It is
easy to see that using the crossing identity on this type of spin
network would delete one edge but also introduce another in the same
circuit, thus the total length of the circuit does not change and
there is no reduction of the graph. In order to decompose this spin
network, it would be necessary to change the orientation of a vertex
by the ``twisting'' operation defined above. This would give an
overall twist factor dependent on the edges involved. This operation
is, however, highly arbitrary since, depending on the choice of the
vertex to be twisted, one obtains either of the embeddings above or
even the original embedding, thus, it is not a viable way to proceed.
Nevertheless, since we were able to identify the toroidal phase factor
with the handle of the torus and obtained (up to orientation of the
surface) a symbol corresponding to this surface, one might ask if
all spin networks embeddable in an orientable closed surface with
genus $>0$ could be expressed as a sum of products of (quantum) $6j$-
and toroidal symbols, one for each handle. It is not hard to imagine
the existence of graphs with such evaluation.

At this point the embedding of the $(3,3)$-bipartite graph in the
double torus is of great interest since it might be the missing link
needed to generalize the Decomposition Theorem for spin networks with
genus $>1$. We could use the inverse operations of the ones used
in Moussouris' algorithm on this embedding in order to introduce enough
vertices and edges such that the resulting graph has two components,
one on each handle. It would then be possible to separate the components
using the generalized Wigner-Eckart theorem, \cite{moussouris1983quantum}.
This could help us to study the possibility of an evaluation of non-planar
graphs as a sum of products of (quantum) $6j$- and toroidal symbols.
In fact, a theorem of Battle, Harary, Kodama and Youngs states that
the orientable genus of a graph is additive over its blocks, \cite{beineke1997graph},
hence, together with the classification of orientable closed surface,
we might expect the above mentioned generalization. 

It might be useful to consider the case where the complete graph on
four vertices $K_{4}$ is toroidal. In this case, the embedding has
a 4-cycle that can be reduced to give a tetrahedron with a vertex
having the wrong orientation. Remember that the $K_{3,3}$ graph is
a minor of the complete graph on five vertices, as the toroidal phase
factor is a minor of $K_{4}$ cellular embedded in a torus, hence,
it might be the case that the embedding of $K_{3,3}$ in the 2-torus
is a ``phase factor of higher genus''. These speculative considerations
are based on some properties of the graphs, however, more work would
be needed in order to understand the existence, if any, of these relations.

However, one has to be careful in making a claim about generalizing
the Decomposition Theorem based only on toroidal symbols and phase
factors, since the existence of forbidden families of graphs prevent
us to rush into such conclusions. Following theorem is the generalization
of Kuratowski's theorem \vref{thm:Kuratowski's-Theorem} for graphs
which are not embeddable in surfaces $S_{h}$ of a given genus $h$,
\cite{beineke1997graph}.
\begin{thm}
For any surface S, the family of homeomorphically minimal graphs that
are not embeddable in S, called the minimal forbidden family $\mathcal{F}(S)$,
is finite.

The minimal forbidden family $\mathcal{F}(S)$ is define to be the
set of graphs having the following three properties: (i) if the graph
$G\in\mathcal{F}(S)$, then $G$ is not embeddable in S; (ii) if a
graph $G$ is not embeddable in S, then $G$ contains a subgraph homeomorphic
to some graph in $\mathcal{F}(S)$; (iii) $\forall G\in\mathcal{F}(S):$
$G$ is not homeomorphic to any subgraph of another graph in $\mathcal{F}(S)$. 
\end{thm}
This theorem might lower our expectations of being able to express
a given graph in terms of toroidal symbols since we expect other spin
networks to be non-toroidal. For instance, there are more than 800
minimal forbidden graphs known for the torus, \cite{beineke1997graph}.
Thus, we need to analyze the general case in order to determine if
the evaluation of all spin networks with genus $>1$ can be expressed
as a sum of products of quantum $6j$- and toroidal symbols. Nevertheless,
using the analogy of the role of $6j$-symbols in the expression of
toroidal symbols the hope remains that the latter could be the factors
needed to evaluate spin networks of higher genus.

\section{Conclusion}

The above analysis of the minimal non-planar spin network $K_{3,3}$
shows that the topology of the spin network, i.e. the orientation
of its vertices and with it the surface in which it is embeddable,
is reflected in its evaluation. This affects the overall behavior
of the coupling symbol and gives rise to new ``symmetry'' properties.
Moreover, three types, or families, of minimal non-planar spin networks
where found. Two of them are toroidal and related by a constant factor
involving exponents of the deformation parameter of the underlying
quantum group. The examples given above of these two families of networks
are representative of each of them, but are not the only ones since
every configuration of the orientation of its vertices might give
rise to another toroidal symbol, or at least proportional to one of
the above, even if they are of the same family. In order to understand
how all of these toroidal symbols relate to each other a short analysis
of the effect on different configurations of orientations for a given
type is needed. 

To account for the non-planar nature of the spin networks discussed
an extension of Moussouri's algorithm was done using the twisting
identity of the Temperley-Lieb recoupling theory. Since only cellular
embeddings were used, it was possible to identify the resulting factor
in the evaluation with the handle of the torus in which the spin networks
were embedded. It was, however, not possible to fully achieve the
generalization of the theorem considered due to the inability of the
algorithm to reduce unambiguously the third type of minimal spin network
found, namely, the one embedded in a 2-torus with only one 2-cell.
Instead some considerations about how to proceed to attempt a generalization
of the theorem and the issues involved were discussed. Some hints
about the possibility of such generalization given by the classification
of closed orientable surfaces and the additive nature of the genus
of a graph are present, however, one has to be careful to conclude
a generalization based on these hints due to the existence of minimal
forbidden families of graphs for a given surface of genus $h$. For
this reason, one would need to consider the general case and extend
the algorithm accordingly, maybe with the generalized Wigner-Eckart
theorem for coupling graphs as given in \cite{moussouris1983quantum},
to account for arbitrary graphs of higher genus.

\subsection*{Acknowledgements: }

I am deeply grateful to John Barrett for having me in Nottingham and
allowing me to gain an understanding of a broad spectrum of very interesting
concepts. I also want to thank the University of Nottingham and the
Erasmus Internship Programm for supporting me during my stay; thanks
as well to the staff at the University of Munich for being so helpful
and patient, especially Mrs. Schleiss and Mr. Emmer. Finally, I want
to thank S. Hofmann for making it possible for me to be in Nottingham.

\bibliographystyle{plainnat}
\bibliography{Bib_Database}

\end{document}